\documentclass[letterpaper, 10 pt, conference]{ieeeconf}  

\IEEEoverridecommandlockouts                              

\usepackage{color}
\usepackage{graphics} 
\usepackage{epsfig} 
\usepackage{times} 
\usepackage{amsmath} 
\usepackage{amssymb}  
\usepackage{tikz}
\usepackage{comment}
\usepackage{cancel}

\newcommand{\R}{\mathbb{R}}
\newtheorem{theorem}{Theorem}[section]

\newtheorem{definition}{Definition}[section]

\newtheorem{example}[definition]{Example}

\newtheorem{remarkth}[definition]{Remark}
\newtheorem{algorithm}[definition]{Algorithm}
\newenvironment{remark}{\begin{remarkth}\upshape}{\end{remarkth}}
\newcommand{\dfb}{\stackrel{\Delta}{=}}

\newcommand{\Flder}{\rightarrow}

\usepackage{amssymb}

\newcommand{\proa}{A^*G \mbox{$\;$}_{\tau^*} \kern-3pt\times_\alpha
G \mbox{$\;$}_\beta \kern-3pt\times_{\tau^*} A^*G}

\newcommand{\lvec}[1]{\overleftarrow{#1}}

\newcommand{\e}{\mbox{exp}}

\newcommand{\al}{\mathfrak{g}}

\title{On the observability of relative positions in left-invariant multi-agent control systems and its application to formation control}

\author{Leonardo Colombo, Hector Garcia de Marina, Mar\'ia Barbero Li\~n\'an and David Mart\'in de Diego
\thanks{L. J. Colombo (leo.colombo@icmat.es) and D. Mart\'in de Diego (d.martin@icmat.es) are with Instituto de Ciencias Matem\'aticas (CSIC-UAM-UCM-UC3M), Calle Nicol\'as Cabrera 13-15, Campus Cantoblanco, 28049, Madrid, Spain. H. Garc\'ia de Marina (hgm@mmmi.sdu.dk) is with Unmanned Aerial Systems Center, The Maersk McKinney Moller Institute, University of Southern Denmark, Odense, Denmark. M. Barbero  Li\~n\'an (m.barbero@upm.es) is with Departamento de Matem\'atica Aplicada, Universidad Polit\'ecnica de Madrid, Av. Juan de Herrera 4, 28040 Madrid, Spain} 
\thanks{This work was partially supported by I-Link 
Project (Ref: linkA20079) from CSIC; Ministerio de Econom\'ia, Industria y
Competitividad (MINEICO, Spain) under grant MTM2016-76702-P; ``Severo Ochoa Programme for Centres of Excellence'' in R$\&$D (SEV-2015-0554) and by a grant from Programa Propio I+D+i de la UPM para J\'ovenes Investigadores Doctores (VJIDOCUPM18MBL).The project that gave rise to these results received the support
of a fellowship from ``la Caixa'' Foundation (ID 100010434). The
fellowship code is LCF/BQ/PI19/11690016.}}%
\begin{document}

\maketitle
\thispagestyle{empty}
\pagestyle{empty}

\begin{abstract}
We consider the localization problem between agents while they run a formation control algorithm. These algorithms typically demand from the agents the information about their relative positions with respect to their neighbors. We assume that this information is not available. Therefore, the agents need to solve the observability problem of reconstructing their relative positions based on other measurements between them. We first model the relative kinematics between the agents as a left-invariant control system so that we can exploit its appealing properties to solve the observability problem. Then, as a particular application, we will focus on agents running a distance-based control algorithm where their relative positions are not accessible but the distances between them are. 
  \end{abstract}

\section{Introduction}
Robot swarms are envisioned to assist humans in challenging tasks such as search\&rescue, disaster response, and environmental monitoring \cite{yang2018grand}.  Formation control algorithms are powerful tools for the control of geometrical variables such as distances, positions or angles between networked robots \cite{oh2015survey}, and they have been found useful for the introduction of robot swarms in real-world scenarios \cite{AnYuFiHe08}. It is inevitable that cost-effective massive robot swarms will require hardware to be as low cost as possible. However, the implementation of many well-studied formation controllers demands robots to know about their relative position with respect to their neighbors, and this information is typically hard to obtain. For example, the relative positions can be estimated directly by employing onboard radars or cameras, which are usually computationally intensive \cite{tron2014distributed,nageli2014environment}. This fact motivates researchers to look for cheaper alternatives, although technically challenging. For example, another works \cite{Guo2018,van2018board} focus on ultra wideband radio chips that can measure distances between them with accuracy in large areas. These radio chips are relatively cheap and, most importantly, very light. Therefore, they are very appealing to be employed in micro aerial vehicles for assisting in relative localization tasks. In this paper, we propose to exploit left-invariant vector fields to assist in the observability problem of reconstructing relative positions from relative geometrical quantities such as inter-agent distances or angles.

Left-invariant vector fields are uniquely determined by only knowing their value at the identity element of the group, which is independent of any particular representation, so it provides a coordinate-free language. These fields enable us to translate the dynamics on the Lie group (a nonlinear space) to its corresponding Lie algebra (a linear vector space). Consequently, the corresponding dynamics for the multi-agent system in this paper will be intrinsic, i.e., they do not depend on the choice of coordinates and permits to avoid the inversion of Jacobians in the proposed framework. This inherent property is a necessity if one desires globally well-defined behavior. We will show how to tackle the non--linear observability problem when a multi-agent system is based on left-invariant vector fields. This approach encompasses other more particular techniques \cite{martinelli2005observability,zelazo2015bearing,williams2015Observability,schiano2018dynamic,colombo2018motion}. For example, we are not restricted to work in absolute coordinates but relative ones, and we are not bound to any specific observation function. In fact, we allow the analysis for the combination of several different sensors or observation functions at each agent. As a particular example, we will study the case where agents only measure inter-agent distances and their impact on distance-based controllers based on rigidity theory \cite{AnYuFiHe08}.

The rigidity theory has been employed before for the relative localization of static agents \cite{aspnes2006theory}. For non-static mobile agents whose position is unknown, we will show how different implementations of the estimation of the relative positions, e.g., with a Kalman filter, in combination with a distance-based formation controller can arise robustness issues in the multi-agent system. We will identify that achieving the desired formation shape does not guarantee a correct estimation of the relative positions. Furthermore, an incorrect shape with an undesired motion of the team of agents can be an attractive steady-state configuration as well. That creates a surprising connection with the recent findings on the robustness of undirected formations \cite{mou2016undirected,MarCaoJa15,li2018formation} so that the issue can be analyzed and potentially solved. Nevertheless, we will show an implementation technique for the correct estimation of the relative positions while agents run a distance-based formation algorithm.

This paper is organized as follows. We overview the relevant concepts on Lie groups and left-invariant control systems in Section \ref{sec: pre}. We connect the introduced concepts to the nonlinear observability problem in Section \ref{sec: obs}. We continue with formalizing the relative kinematics between neighboring agents in a network as a left-invariant control system in Section \ref{sec: left}. This framework will allow us to study in Section \ref{sec: loc} the reconstruction of relative positions from the observations of distances between neighboring agents. We analyze in Section \ref{sec: rig} some robustness issues on employing the proposed observer in combination with a distance-based formation controller. We illustrate our findings with some numerical experiments in Section \ref{sec: exp}.

\section{Preliminaries}
\label{sec: pre}
We will consider that all manifolds in this section are $C^{\infty}$, and we employ Einstein's summation convention, i.e., we drop the summation sum over repeated indices.
\subsection{Elements of differential geometry}
Let $Q$ be an $m$-dimensional manifold where its \emph{tangent bundle} $TQ$ is the collection of all the tangent vectors to $Q$ at each point.
The tangent bundle projection $\tau_{TQ}:TQ\Flder Q$ assigns to each tangent vector its base point.


Let $f: Q\Flder N$ be a smooth mapping between manifolds $Q$ and $N.$ We then write $Tf:TQ\Flder TN$ to denote the tangent map.  When $N=\R$ we shall denote the set of smooth real-valued functions on $Q$ by $C^{\infty}(Q).$
Consider the linear map $\mathbf{d}f(q):T_{q}Q\to\mathbb{R}$ called \textit{differential} of $f$ and $\mathbf{d}f(q)\in T^*_qQ$, the dual of the vector space $T_qQ$ . For $v_q\in T_{q}Q$, $\mathbf{d}f(q)\cdot v_q$ provides the directional derivative of $f$ and it is locally given by $\mathbf{d}f(q)\cdot v_q=\frac{\partial f}{\partial q^{i}}v_q^{i}$ where $(q^{i})$ are local coordinates on $Q$.

A \textit{vector field} $X$ on $Q$ is a smooth mapping $X:Q\Flder
TQ$ which assigns to each point $q\in Q$ a tangent vector $X(q)\in
T_{q}Q$ and it satisfies $\tau_{TQ}\circ X={\rm Id}_{Q}$, where ${\rm Id}_{Q}$ is the identity map on $Q$ and $\tau_{TQ}:TQ\to Q$ the canonical projection, $\tau_{TQ}(q,v_{q})=q$. The set of all vector fields
over $Q$ is denoted by $\mathfrak{X}(Q).$ An \textit{integral curve}
of a vector field $X$ is a curve satisfying $\dot{c}(t)=X(c(t)).$

Let $X\in \mathfrak{X}(Q)$ and $h\in C^{\infty}(Q)$. The Lie derivative of $h$ with respect to $X$ is the real-valued function $\mathcal{L}_{X}h:Q\to\mathbb{R}$ given by $\mathcal{L}_{X}h(q):=\textbf{d}h(q)\cdot X(q)=X^{i}\frac{\partial h}{\partial q^{i}},\,\forall q\in Q,$
using as local coordinates $(q^i)$ on $Q$.

If $h:Q\to\mathbb{R}^{p}$ is a  vector valued differentiable function, applying the above definition component-wise we have $\mathcal{L}_{X}h(q):=\Big{[}X^{i}\frac{\partial h_1}{\partial q^{i}},\ldots,X^{i}\frac{\partial h_p}{\partial q^{i}}\Big{]}^{T}=X(q)\frac{\partial h}{\partial q},$ in which case $\mathcal{L}_{X}h:Q\to\mathbb{R}^{p}$. 






\subsection{Lie group and left-invariant vector fields}

Next, we introduce the basics on Lie groups and left-invariant vector fields (see \cite{thesisleo} Chapter $2$ for more details).




Let $G$ be a finite dimensional Lie group. The tangent bundle at a point $g\in G$ is denoted as $T_{g}G$ . The tangent space at the identity $e\in G$ defines a Lie algebra and is denoted by $\mathfrak{g}:=T_{e}G$.
 Let $L_{g_1}:G\to G$ be the left translation of the element $g_1\in G$ given by $L_{g_1}(g_2)=g_1g_2$ for $g_2\in G$. $L_{g_1}$ is a diffeomorphism on $G$. Their tangent map  (i.e, the linearization or tangent lift of left translations) is denoted by $T_{g_2}L_{g_1}:T_{g_2}G\to T_{g_1g_2}G$. 




\begin{definition}\label{defLI} $X\in\mathfrak{X}(G)$ is called \textit{left-invariant} if $T_{g_2}L_{g_1}(X(g_2))=X(L_{g_1}(g_2))=X(g_1g_2)$ $\forall\,g_1,g_2\in G$.\end{definition}

In particular for $g_1=g$ and $g_2=e$, Definition \ref{defLI} means that a vector field $X$ is left-invariant if $\dot{g}=X(g)=T_{e}L_{g}\xi$ for $\xi=X(e)\in\mathfrak{g}$. As $X$ is left-invariant, $\xi=X(e)=T_{g}L_{g^{-1}}\dot{g}$. The tangent map $T_{e}L_g$ shifts vectors based at $e$ to vectors based at $g\in G$. By doing this operation for every $g\in G$ we define a vector field as $\displaystyle{\lvec{\xi}}(g):=T_{e}L_g(\xi)$ for $\xi:=X(e)\in T_{e}G$. Note that the vector field $\lvec{\xi}(g_1)$ is left-invariant, because $\lvec{\xi}(g_2g_1)=T_e(L_{g_2}\circ
L_{g_1})\xi=T_{g_1}L_{g_2}\lvec{\xi}(g_1).$

From now on, the left arrow above a vector field will denote it is left-invariant. The set of left-invariant vector fields on $G$ is isomorphic to $T_eG$ as vector spaces. Thus, the left-invariant vector fields on $G$ are uniquely determined by knowing their value at the identity element as we shown above. Moreover, the set of left-invariant vector fields is a Lie algebra, i.e., the Lie bracket of left-invariant vector fields is left-invariant: $[\lvec{\xi},\lvec{\eta}]=\lvec{[\xi,\eta]}$. 

For all $\xi\in\al$, let $\gamma_{\xi}:\R\Flder G$ denote the integral curve of the left-invariant vector field $\lvec{\xi}$ induced by $\xi$, which is defined uniquely by claiming
\[
\lvec{\xi}(e)=\xi,\,\,\,\,\gamma_{\xi}(0)=e,\,\,\,\,\gamma_{\xi}^{\prime}(t)=\lvec{\xi}(\gamma_{\xi}(t))\,\,\,\mbox{for all }\,t\in\R.
\]
\begin{definition}\label{exponencial}
The map $\e:\al\Flder G,\,\,\,\,\e(\xi)=\gamma_{\xi}$ is called the \emph{exponential map}.
\end{definition}

By using the exponential map, a left-invariant vector field $\lvec{\xi}$ induced by $\xi$ may be constructed as
\begin{equation}
\lvec{\xi}(g)=\left.\dfrac{\rm d}{{\rm d}t}\right|_{t=0}(g\, {\rm exp}(t\xi)).
	\label{eq: eb}
\end{equation}

\subsection{Left-invariant control systems}
Let $G$ be an $m$-dimensional Lie group.
\begin{definition}[see \cite{jurdjevic}]  A control system $\dot{g}=f(g,u)$, defined on $G$,  where $u\in\mathfrak{g}$ is said to be \textit{left-invariant} if $T_{g_2}L_{g_1}f(g_2,u)=f(L_{g_1}(g_2),u)$ for each $g_1,g_2\in G$.\end{definition}

For a left-invariant fully actuated control system such that the Lie algebra $\mathfrak{g} $ is spanned by $\{E_{1},\hdots,E_{m}\}$, the controls are specified for the generators of the Lie algebra as $\displaystyle{T_gL_{g^{-1}}\dot{g}=u=\sum_{i=1}^{m}u^{i}(t)E_{i}}$,

\begin{definition}[see \cite{jurdjevic}]  A \textit{left-invariant control system} on a Lie group $G$ is given by $\dot{g} = T_{e}L_{g}(u)$.\end{definition}
Each left-invariant control system is defined by its values at the identity $e$ of the group, since $T_eL_gf(e,u)=f(g,u)$. 
%
%

\begin{remark}
In most of the applications controls are used to externally influence the angular/linear velocity of rigid bodies, so it is helpful to keep in mind that $u$ may play the role of angular velocity for a rigid
 body on the body frame.
\end{remark}

\section{Nonlinear observability for left-invariant control systems}
\label{sec: obs}
Consider the following left-invariant control system
\begin{equation}
\label{acs}
\dot g = gu = g\sum_{j=1}^{m}u^jE_j,\quad y=h(g),
\end{equation} 
where $h:G\to\mathbb{R}^{p}$ is an output map.

Let $g_1$, $g_2$ be two points in an open set $V$ of $G$. They are $V$-indistinguishable 
if for every admissible constant control $u:[0,T]\to U$ the corresponding integral curves starting from $g_1$ and $g_2$ remain in $V$ for all $t\leq T$ and the output function are the same for both trajectories for all $t\leq T$.

\begin{definition}
The control system \eqref{acs} is said to be locally observable if every state $g$ can be distinguished from its neighbors using system trajectories remaining close to $g$. 
\end{definition}

\begin{definition}
	The observation space $\mathcal{O}$ for the control system \eqref{acs} is the space of functions on $G$ containing $h_1,\ldots,h_p$ and all iterated Lie derivatives of left-invariant vector fields $\mathcal{L}_{\lvec{X_1}}\mathcal{L}_{\lvec{X_2}}\ldots \mathcal{L}_{\lvec{X_s}}h_j,$ $j\in \{1,\dots,p\}$, $s\in\mathbb{N}$ where $\lvec{X_s}$ are the left-invariant vector fields associated with the element $E_s$ of the basis of $\mathfrak{g}$, e.g., see equation (\ref{eq: eb}).
\end{definition}

Roughly speaking, $\mathcal{O}$ contains all the output functions and all derivatives of the output functions along trajectories of the system. Given the observation space $\mathcal{O}$ and $g\in G$, the observability is determined by studying the dimension of the following space which represents the space of feasible observations: $\hbox{d}\mathcal{O}(q)=\hbox{span}\{\hbox{d}\alpha(g)|\alpha\in\mathcal{O}\}.$

\begin{theorem}
The left-invariant control system \eqref{acs} is locally observable at $g\in G$ if $\dim\{\hbox{d}\mathcal{O}(g)\}=\dim\{G\} $.
\label{th: obs}
\end{theorem}
\begin{proof}
It follows the same lines as  the proof in \cite{van} (Theorem 3.32) for control affine systems.
\end{proof}

\section{Left-invariant multi-agent systems}
\label{sec: left}

\subsection{Agents in a network}
Consider a set $\mathcal{N}$ of $o\in\mathbb{N}\geq2$ agents whose position in the plane is denoted by $r_i\in\mathbb{R}^2, i\in\{1,\dots,o\}$ with respect to a fixed global frame, and define $r=(r_1,\dots,r_o)\in\mathbb{R}^{2o}$ as the stacked vector of agents' positions.
An agent $i\in\mathcal{N}$ can take measurements with respect to other agents in the subset $\mathcal{N}_i \subseteq \mathcal{N}$, i.e., the neighbors of agent $i\in\mathcal{N}$. The neighbor relationships are described by an undirected graph $\mathbb{G} = (\mathcal{N}, \mathcal{E})$ with the edge set $\mathcal{E}\subseteq\mathcal{N}\times\mathcal{N}$. The set $\mathcal{N}_i$ is defined by $\mathcal{N}_i\dfb\{j\in\mathcal{N}:(i,j)\in\mathcal{E}\}$. 
We define the elements of the incidence matrix $B\in\mathbb{R}^{o\times|\mathcal{E}|}$ that establish the neighbors' relationships for $\mathbb{G}$ by
\begin{equation}
	b_{ik} \dfb \begin{cases}+1 \quad \text{if} \quad i = {\mathcal{E}_k^{\text{tail}}} \\
		-1 \quad \text{if} \quad i = {\mathcal{E}_k^{\text{head}}} \\
		0 \quad \text{otherwise}
	\end{cases},
	\label{eq: B}
\end{equation}
where $\mathcal{E}_k^{\text{tail}}$ and $\mathcal{E}_k^{\text{head}}$ denote the tail and head nodes, respectively, of the edge $\mathcal{E}_k$, i.e., $\mathcal{E}_k = (\mathcal{E}_k^{\text{tail}},\mathcal{E}_k^{\text{head}})$.

The stacked vector of relative positions between neighboring agents is then given by
\begin{equation}
	z = \left(\begin{matrix} \overline B^T \\ -\overline B^T \end{matrix}\right) r,
	\label{eq: z}
\end{equation}
where $\overline B := B \otimes I_2$ with $I_2$ being the $2\times 2$ identity matrix, and $\otimes$ the Kronecker product. Note that $z_k \in \mathbb{R}^2$ and $z_{k+|\mathcal{E}|}\in\mathbb{R}^2$ in $z$ correspond to $r_i - r_j$ and $r_j - r_i$ for the edge $\mathcal{E}_k$ respectively. We can also define $r_{ij} \dfb r_i - r_j$ with respect to a global frame to reduce verbosity.

\subsection{Relative kinematics in the network}
In order to estimate the relative position between neighboring agents, we will focus on the relative motions, or velocities, of the neighbors of agent $i$ with respect to agent $i$. Denote by $SO(2n)$ the orthogonal group of dimension 
$2n$. The mathematical construction needed is related to the Lie group
$SE(2n)=\{(p,R): p\in \mathbb{R}^{2n}, R\in SO(2n)\}\;,$
with $n\in\mathbb{N}\geq 1$. An element of $SE(2n)$ is usually represented in matrician form for operational purposes as 
${\bf q}=\left(
\begin{array}{cc}
R&p\\
0&1
\end{array}
\right)$, where the multiplication on $SE(2n)$ and inverse are
\[
{\bf q}\cdot {\bf q}'=\left(
\begin{array}{cc}
R&p\\
0&1
\end{array}
\right)
\left(
\begin{array}{cc}
R'&p'\\
0&1
\end{array}
\right)
=
\left(
\begin{array}{cc}
RR'&Rp'+p\\
0&1
\end{array}
\right)\, ,
\]
\[
{\bf q}^{-1}=\left(
\begin{array}{cc}
R^T&-R^Tp\\
0&1
\end{array}
\right).
\]
As we will employ distance-based formation controllers in Section \ref{sec: rig}, the agent $i$ will work in its own local frame of coordinates to measure the distances to all the $n_i = |\mathcal{N}_i|$ neighbors. Therefore, for every agent $i$ we will have interest in the following Lie subgroup of $SE(2n_i)$
\[
G_i=\{(p, R_i) \; |\; p\in {\mathbb R}^{2n_i} ,\ R_i\in SO(2) \},
\]
where $p=(p_1,\dots,p_{n_i}) $ is the stacked vector such that $p_k$ represents the relative position of the neighbor $k$ in $\mathcal{N}_i$ with respect to $i$ and $R_i$ is the rotational matrix representing the orientation of agent $i$ with respect to a global frame of coordinates. The Lie subgroup $G_i$ is embedded in $SE(2n_i)$ as
$
(p, R_i)\hookrightarrow (p,R)= (p, \underbrace{R_i,\ldots,R_i}_{\text{$n_i$-times}}),
$
or in matricial form
\[
\left(
\begin{array}{cc}
R&p\\
0&1
\end{array}
\right)=
\left(
\begin{array}{ccccc}
R_i\otimes I_{n_i} &p\\
0&1
\end{array}
\right).
\] 
The multiplication in the Lie subgroup $G_i$ is defined by 
\[
	(p, R_i)(p', R'_i)=\left(\left(R_ip'_k+p_k\right)_{1\leq k\leq n_i}, R_iR_i'\right).
\]
where ${p}_k \in\mathbb{R}^2$ is the $k$'th vector in the stacked vector $p$. Once we center the system at the agent $i$, a point ${\mathbf q}\in G_i$ carries the information of the relative position of all the neighbors ${\mathcal N}_i$ with respect to the agent $i$ and the orientation of the agent $i$. The point  ${\mathbf q}$ can be represented using (homogeneous) coordinates $(p,\theta_i)=(x_1, y_1, \ldots, x_{n_i}, y_{n_i}, \theta_i)$ where
$p$ is the stacked vector of the relative positions $r_{ij}$ such that $j\in\mathcal{N}_i$ and
$
R_i(\theta_i)=
\left(
\begin{array}{cc}
\cos\theta_i&-\sin \theta_i\\
\sin \theta_i&\cos\theta_i	
\end{array}
\right)
$
denotes the local frame of coordinates of agent $i$ with respect to a fixed global frame. Denote by $\mathfrak{g}_i$ the Lie subalgebra of  $G_i$. The Lie subalgebra $\mathfrak{g}_i$ of the Lie algebra $\mathfrak{se}(2n_i)$, as a vector space, is given by 
$\mathfrak{g}_i=\{\xi=(v, \omega_i)\; |\; v\in {\mathbb R}^{2n_i}, \omega_i\in \mathfrak{so}(2)\}
$,
where 
$\omega_i=\left(
\begin{array}{cc}
0&-w_i\\
w_i&0	
\end{array}
\right)$ with $w_i\in\mathbb{R}$ denoting the angular velocity of agent $i$, and $v$ is the stacked vector of relative velocities between agent $i$ and the neighbors ${\mathcal N}_i$. Alternatively we can embed the Lie algebra
$\mathfrak{g}_i=\{\xi=(v, w_i)\; |\; v\in {\mathbb R}^{2n_i},  w_i\in {\mathbb R}\}
$ into $\mathfrak{se}(2n_i)$ by taking
$\xi=
\left(
\begin{array}{cc}
\omega_i\otimes I_{2n_i}&v \\
0&0
\end{array}
\right).
$

Let $\{e_{l}\}_{1\leq l\leq 2n_i}$ be the canonical basis of ${\mathbb R}^{2n_i}$,  the Lie subalgebra $\mathfrak{g}_i$ has  $(2n_i+1)$-generators $\{E^a_{k}, E_{n_i+1}\}$,  $1\leq a\leq 2$ and $1\leq k\leq n_i$, with $E^a_{k}=( e_{a+k}, 0)$ and $E_{n_i+1}=(0, \ldots, 0,  1)$. The Lie bracket of the generators satisfies
$
[E^a_{k}, E^b_{j}]=0, \quad [E^1_{k}, E_{n_i+1}]=E^2_{k}, \quad [E^2_{k}, E_{n_i+1}]=-E^1_{k}\;$, and the exponential map on $G_i$, in coordinates, is given by
\begin{eqnarray*}
\hbox{exp}(v, \omega_i)=\left(
\begin{tiny}\begin{array}{ccccc}
	R_i(\omega_i)&0&\cdots & 0 &R_i(\omega_i)\frac{v^\perp_{1}}{\omega_i}-\frac{v^\perp_{1}}{\omega_i}\\
	0&R_i(\omega_i)&\cdots &0&R_i(\omega_i)\frac{v^\perp_{2}}{\omega_i}-\frac{v^\perp_{2}}{\omega_i}\\
	\cdots& \cdots   &\cdots &\cdots &\cdots\\
	0       & 0      &\cdots      & R_i(\omega_i)     & R_i(\omega_i)\frac{v^\perp_{n_i}}{\omega_i}-\frac{v^\perp_{n_i}}{\omega_i}\\
	0       & 0       &\cdots     & 0        & 1
\end{array}\end{tiny}
\right)
\end{eqnarray*}
$\omega_i\not=0$, where
$
R_i(\omega_i)=\left(
\begin{array}{cc}
\cos\omega_i&-\sin \omega_i\\
\sin \omega_i&\cos\omega_i	
\end{array}
\right)
$ and if $v=(v_x, v_y)\in {\mathbb R}^2$ then $v^\perp=(v_y, -v_x)$.
Moreover, $\hbox{exp}(v, 0)=(v, I_2)$ and
the left-invariant invariant vector fields corresponding to the basis  $\{E^a_{k}, E_{n_i+1}\}$  in coordinates $(x_1, y_1, \ldots, x_{n_i}, y_{n_i}, \theta_i)$  are
\begin{align*}
		\lvec{{E}_{n_i+1}}&=\frac{\partial}{\partial \theta_i},\quad \lvec{{E}^1_{k}}=\cos\theta_i\frac{\partial}{\partial x_k}+\sin\theta_i\frac{\partial}{\partial y_k} \\
	\lvec{{E}^2_{k}}&=-\sin\theta_i\frac{\partial}{\partial x_k}+\cos\theta_i\frac{\partial}{\partial y_k}, \quad k\in\{1,\dots,n_i\}
\end{align*}
For a curve ${\mathbf q}: t\rightarrow G_i$ we have that
\[
{\mathbf q}^{-1}\dot{\mathbf q}=
\left(
\begin{array}{ccccc}
	R_i^{-1}\dot{R}_i&0&\cdots & 0 &-R_i^T\dot{p }_1\\
	0&R_i^{-1}\dot{R}_i&\cdots &0&-R_i^T\dot{p}_2\\
\cdots& \cdots   &\cdots &\cdots &\cdots\\
0       & 0      &\cdots      & R_i^{-1}\dot{R}_i  & -R_i^T\dot{p}_{n_i}\\
0       & 0       &\cdots     & 0        & 0
\end{array}
\right),
\] implying that ${\mathbf q}\in G_i$ satisfies $\dot{\mathbf q}={\mathbf q}\, \xi,$ so we can write
\begin{equation}\label{systemlics}
	\begin{cases}
\dot{p}_{j}&= R_i\,v_{j}\\
\dot{R}_i&= R_i\,\omega_i,
\end{cases}
\end{equation}
	which is a fully actuated system on the Lie subalgebra $\mathfrak{g}_i$ of $\mathfrak{se}(2n_i)$ with control inputs $(v_j, w_i)$,  where $v_{j}$ is the relative velocity between the agent $i$ and the neighbor agents $j$. Therefore, in terms of the basis of $\mathfrak{g}_i$ we write can write the dynamics of $\mathbf{q}$ in the following compact form
	\begin{equation}
		\dot{\mathbf q}= w_i\lvec{{E}_{n_i+1}}+\sum_{k=1}^{n_i}\left( v_{x,k} \lvec{{E}^1_{k}}+v_{y.k} \lvec{{E}^2_{k}}\right). \label{eq: dynq}
\end{equation}

\section{Relative localization of agents with distance measurements}
\label{sec: loc}
We consider that agents have installed on board sensors that enable agent $i$ to measure its distance from agent $k\in {\mathcal N}_i$ and its own orientation $\theta_i$ with respect to a global frame of coordinates. From now on we focus on the agent $i$, for $k=1,\dots,n_i$ we define the observation functions on $G_i$
	\begin{equation}
	 h_{k}({\mathbf q})=\frac{1}{2}\left(x_k^2+y_k^2\right), \quad h_{n_i+1}({\mathbf q})=\theta_i.\label{eq: obs}
	\end{equation}
\begin{theorem}
	\label{th: main}
	Assume that all the neigbors of agent $i$  are in relative motion with respect to agent $i$ in both coordinates $x$ and $y$. Then the state $\mathbf{q}$ in $G_i$  is observable under the dynamics (\ref{eq: dynq}) with the observation functions $h_k$ and $h_{n_i+1}$ for $k=1,\dots,n_i$.
\end{theorem}
\begin{proof}
To construct the observability matrix we first compute:
\begin{align*}
\quad dh_{k}&= x_k\, dx_k+ y_k\, dy_k, \quad dh_{n_i+1}= d\theta_i,
\\
	d (\lvec{{E}^1_{j}} h_{k})&=\delta_{jk}\left(\cos\theta_i\,  dx_k  + \sin\theta_i\, dy_k\right.\\& \left.-(x_k\sin \theta_i-y_k\cos\theta_i)\, d\theta_i\right), \\
 d (\lvec{{E}^2_{j}} h_{k})&=\delta_{jk}\left(-\sin\theta_i\,  dx_k  + \cos\theta_i\, dy_k\right.\\&\left.-(x_k\cos \theta_i+y_k\sin\theta_i)\, d\theta_i\right)\\
d (\lvec{{E}^1_{j}} h_{n_i+1})&=0, \quad 
d (\lvec{{E}^2_{j}} h_{n_i+1})=0,\\
d (\lvec{{E}_{n_i+1}} h_{k})&=0,\quad
d (\lvec{{E}^1_{n_i+1}} h_{n_i+1})=0,
	\end{align*}
	where $\delta_{jk}$ is the Kronecker delta and  $j,\, k \in\{1,\dots,n_i\}$. We are ready now to calculate the rank of 
\begin{align*}
d\mathcal O(\mathbf{q})&=\{ dh_{k},  dh_{n_i+1}, d (\lvec{{E}^1_{j}} h_{k}), d (\lvec{{E}^2_{j}} h_{k}), d (\lvec{{E}^1_{j}} h_{n_i+1}),\\& d(\lvec{{E}^2_{j}} h_{n_i+1}), d(\lvec{{E}_{n_i+1}} h_{k}), d (\lvec{{E}^1_{n_i+1}} h_{n_i+1}), \ldots\}_{\mathbf{q}}, \nonumber
\end{align*}
in particular, it is easy to check that
\begin{small}\begin{align}
	&\operatorname{rank }d\mathcal O(\mathbf{q}) =  \nonumber \\
	&=\operatorname{rank}\{ d\theta_i,\cos\theta_i\,  dx_k  + \sin\theta_i\, dy_k, -\sin\theta_i\,  dx_k  + \cos\theta_i\, dy_k \} \nonumber \\&=2n_i+1 = \operatorname{dim}\{G_i\}, \quad k\in\{1,\dots,n_i\},
	\label{eq: dom}
\end{align}\end{small}
	Then, the system is observable at each point by Theorem \ref{th: obs}. In particular, the observation of $\theta_i$ is trivial because $h_{n_i+1}(\mathbf{q}) = \theta_i$. However, the observation of $p_{j}$ is possible because it is assumed that $v_{j}$ is not zero in any of the components of $\mathbb{R}^2$. Note that according to  (\ref{systemlics}), when the relative velocity $v_{j}$ between the agent $i$ and the neighbors is zero in one of the components, the pair of neighboring agents makes a parallel translational motion along that component. 
\end{proof}
\begin{remark}
Note that for a rigid rotational motion of two neighboring agents with respect to a fixed point we will have a constant $h_{k}$. However, the corresponding relative position is observable since its associated $v_{j}$ is not zero. This fact will play an important role in distance-based formation control where rotations of the desired shape are allowed. This is not the case when other approaches are taken, for example, in position-based formation control \cite{oh2015survey} where the agents are controlling orientations as well.
\end{remark}
\begin{remark}
The analysis of the rank of $d\mathcal O$ in Theorem \ref{th: main} does not present difficulties because the observation functions $h_{k}$ allow us to look at each pair of neighbors separately. The presented mathematical framework enables us to study such cases even in 3D. However, this would not be the case for different observation functions involving more than one pair of neighbors like the ones suggested in the recent work on \emph{weak rigidity} \cite{park2017rigidity}. 
\end{remark}
\begin{remark}
Note that the fact that the agent $i$ can be self-rotating and this will not interfere with the estimation of its relative positions with respect to its neighbors.
\end{remark}

\section{Distance-based formation control with relative positions estimated from distance measurements}
\label{sec: rig}
Formation control algorithms provide tools to solve the task of forming a particular geometrical shape by a team of agents. In particular, rigidity theory \cite{AnYuFiHe08} allows the description of such shapes by setting desired inter-agent distances. A popular algorithm for each agent $i$ based on the gradient descent technique for minimizing distance errors between agents \cite{MaJaCa15} is given by
\begin{equation}
	\dot r_i = -\sum_{j\in\mathcal{N}_i}r_{ij}e_{ij},
	\label{eq: rdis}
\end{equation}
where $e_{ij} := ||r_{ij}||^2 - d_{ij}^2$ with $d_{ij}\in\mathbb{R}^+$ is the desired distance between agents $i$ and $j$. Let us split the vector (\ref{eq: z}) as $z = (z_1,z_2)^T$, which obviously satisfies $z_1 = -z_2$. Then the control action (\ref{eq: rdis}) can be written in compact form as
\begin{equation}
	\dot r = -R(z_1)^Te,
	\label{eq: rcom}
\end{equation}
where $r\in\mathbb{R}^{2|\mathcal{N}|}$ and $e\in\mathbb{R}^\mathbb{|\mathcal{E}|}$ are the stacked vectors of agents' positions and error distances respectively, and $R$ is the rigidity matrix \cite{AnYuFiHe08}. In particular $R:=\operatorname{diag}\{z_1\}^T(B\otimes I_2)^T$, where the operator $\operatorname{diag}$ place the stacked vectors (not the scalar elements) in $z_1$ in a block diagonal matrix, and we recall that $B$ is as in (\ref{eq: B}). A theoretical analysis of (\ref{eq: rcom}) reveals that desired shapes are locally exponentially stable \cite{KrBrFr08}. However, such a convergent result considers the common theoretical assumption where $z_1 = -z_2$. While it is true that $r_{ij} = -r_{ji}$, robots measure or estimate their corresponding relative positions on board. Therefore, it is realistic to assume that the estimations by agents $i$ and $j$ on the relative position between them are different, i.e., $\hat r_{ij} \neq - \hat r_{ji}$ where we denote by $\hat{\cdot}$ the estimation of the variable. Consequently, if robots implement the algorithm (\ref{eq: rdis}) but with estimated quantities, then
\begin{equation}
	\dot r_i = -\sum_{j\in\mathcal{N}_i}\hat r_{ij}e_{ij},
	\label{eq: rdis2}
\end{equation}
cannot be rewriten as (\ref{eq: rcom}) since $\hat r_{ij} \neq - \hat r_{ji}$. In other words, the stability of the formation based on the properties of the rigidity matrix might be lost. Without loss of generality, let us analyze three practical robustness issues through an example.
\begin{example}
	Consider a team of three agents, whose incidence matrix $B$ describes a complete graph, and their desired distances are all equal to $d$. We assume that each agent only measures distances with respect to its neighbors, and they estimate their relative positions in a Kalman filter with model dynamics (\ref{systemlics}), or alternatively (\ref{eq: dynq}), and observations as in (\ref{eq: obs}). Note that for implementing (\ref{systemlics}) agents need to communicate their velocities with their neighbors, and we consider that all the agents have the same orientation $\theta$. Because agents are measuring inter-agent distances, then they can measure directly their error distances. The implementation of (\ref{eq: rdis2}) at each agent is then given by
\begin{equation}
	\begin{cases}
		\dot r_1 &= -\hat r_{12}e_{12} - \hat r_{13}e_{13} \\
		\dot r_2 &= -\hat r_{21}e_{21} - \hat r_{23}e_{23} \\
		\dot r_3 &= -\hat r_{31}e_{31} - \hat r_{32}e_{32}.
	\end{cases}
\end{equation}
\end{example}

\emph{Robustness Issue 1:} Consider the situation when $\hat r_{12}$ and $\hat r_{13}$ are parallel. Therefore, there exists $e_{12}$ and $e_{13}$ where $\dot r_1 = 0$. If we consider similar cases for all the agents, then there exist values of the estimations where all the agents get stuck, and according to Theorem \ref{th: main} the estimations will not be updated. Therefore, the system reaches an undesired equilibrium with incorrect estimations $\hat r_{ij}$ and an undesired steady-state shape.

Note that the theoretical analysis in \cite{mou2016undirected} guarantees that for initial conditions where the relative positions of the three agents are not parallel, then they will remain non-parallel under dynamics (\ref{eq: rdis}). However, when an estimator like the proposed Kalman filter is running, we can find infinite initial conditions for the estimator such as the Robustness Issue 1 holds. In fact, we have confirmed numerically that the agents can converge to such a configuration even when they start with non-parallel estimations on the relative positions.

One could argue that $e_{ij} \approx e_{ji}$ in order to limit the conditions for the Robustness Issue 1 to happen. However, even considering that $e_{ij} = e_{ji}$ we have that
\begin{equation}
	\dot r = -\begin{pmatrix}\hat r_{12} & 0 & \hat r_{13} \\ \hat r_{21} & \hat r_{23} & 0 \\ 0 & \hat r_{32} & \hat r_{31}\end{pmatrix}\begin{pmatrix}e_{12} \\ e_{23} \\ e_{13} \end{pmatrix},
	\label{eq: eeq}
\end{equation}
where we note that the transition matrix has dimension $6\times 3$ with a non-trivial kernel.

\emph{Robustness Issue 2:} It can be checked that for $\dot r= c\mathbf{1}\in\mathbb{R}^6$ with $c\in\mathbb{R}$ the system (\ref{eq: eeq}) has a solution. Therefore, all the agents will move in pure translation with no necessarily $e_{ij}=0$. Then, according to Theorem \ref{th: main}, the estimations $\hat r_{ij}$ will not be updated. Note that not only the formation is in an undesired shape but it will drift away with velocity $c\mathbf{1}$. We have also confirmed numerically that the agents can converge to such a scenario.

\emph{Robustness Issue 3:} Since the error distances are directly measured, if the agents are at the desired inter-agent distances, then $\dot r = 0$ in (\ref{eq: rdis2}). If the initial conditions for the estimations of $\hat r_{ij}$ are close to the actual values, it might happen that the inter-agent distances converge faster to their desired values than the estimations, i.e., once the formation achieves the desired shape, it will stop moving but with an incorrect estimation of $r_{ij}$. How the system (\ref{eq: rdis2}) can converge to the desired shape once $\hat r_{ij}(t) \approx r_{ij}(t)$ can be explained by the similar scenario analyzed in \cite{bishop2015distributed}.

The consequences of these three robustness issues resemble to the ones described in \cite{mou2016undirected,de2017taming,meng2017three} but triggered by different causes. In order to avoid these problems in a general multi-agent system running a distance-based formation control algorithm, the evident goal is to recover the structure in (\ref{eq: rdis}). We propose the following three-steps algorithm.

\begin{algorithm}
\emph{1st step:} Only one agent per pair of neighboring agents will estimate $r_{ij}$ and share it with its neighbor. 

\emph{2nd step:} The initial conditions for $\hat r_{ij}$ must be close enough to the actual values in order to guarantee convergence to the desired shape as explained in \cite{bishop2015distributed}. 

\emph{3rd step:} In order to guarantee the convergence of $\hat r_{ij}$ to the actual values we force a steady-state rotational motion in the formation to satisfy the condition in Theorem \ref{th: main}. This can be done with the technique proposed in \cite{MaJaCa15} by adding mismatches (a constant number) to the error signals.
\label{al: 1}
\end{algorithm}

For example, in the scenario described in Example 6.1 and after following the three steps in Algorithm \ref{al: 1}, the implementation of the distance-based controller with estimated relative positions from distance measurements such that the desired equilateral triangle is achieved and the relative positions are estimated correctly is given by
\begin{equation}
	\begin{cases}
		\dot r_1 &= -\hat r_{12}(e_{12}-a) - \hat r_{13}(e_{13}-a)\\
		\dot r_2 &= \hat r_{12}(e_{12}+a) - \hat r_{23}(e_{23}-a) \quad a\in\mathbb{R}\\
		\dot r_3 &= \hat r_{13}(e_{13}+a) + \hat r_{23}(e_{23}+a).
	\end{cases}
	\label{eq: ex}
\end{equation}

\section{Numerical experiments}
\label{sec: exp}
We spread randomly the three agents on the plane under the dynamics (\ref{eq: ex}). We set the target distance $d = 10$ and the initial conditions for the estimators in the Kalman filter to an arbitrary number within $\pm 2$ the actual values in both coordinates $x$ and $y$. We set $a = 1$ in (\ref{eq: ex}) for forcing the rotational motion of the formation so that we satisfy the conditions in Theorem \ref{th: main}. The Figures \ref{fig: pos} and \ref{fig: est} show the correct convergence of the distances and the estimation errors to the desired values. In particular, as predicted, the steady-state rotational motion of the formation assists in the correct estimation of the relative positions between the agents.

\begin{figure}
\centering
	\includegraphics[width=0.49\columnwidth]{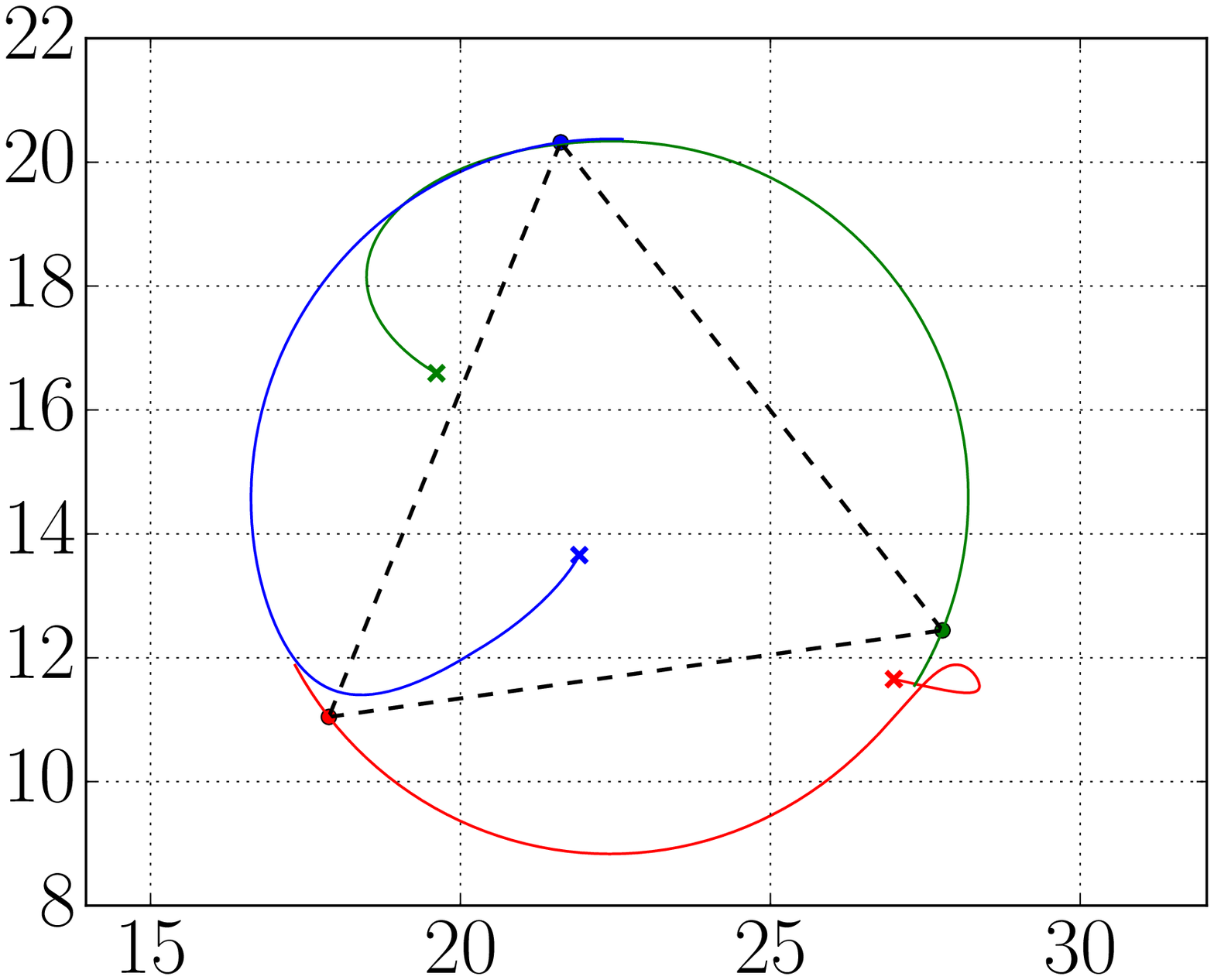}
	\includegraphics[width=0.43\columnwidth]{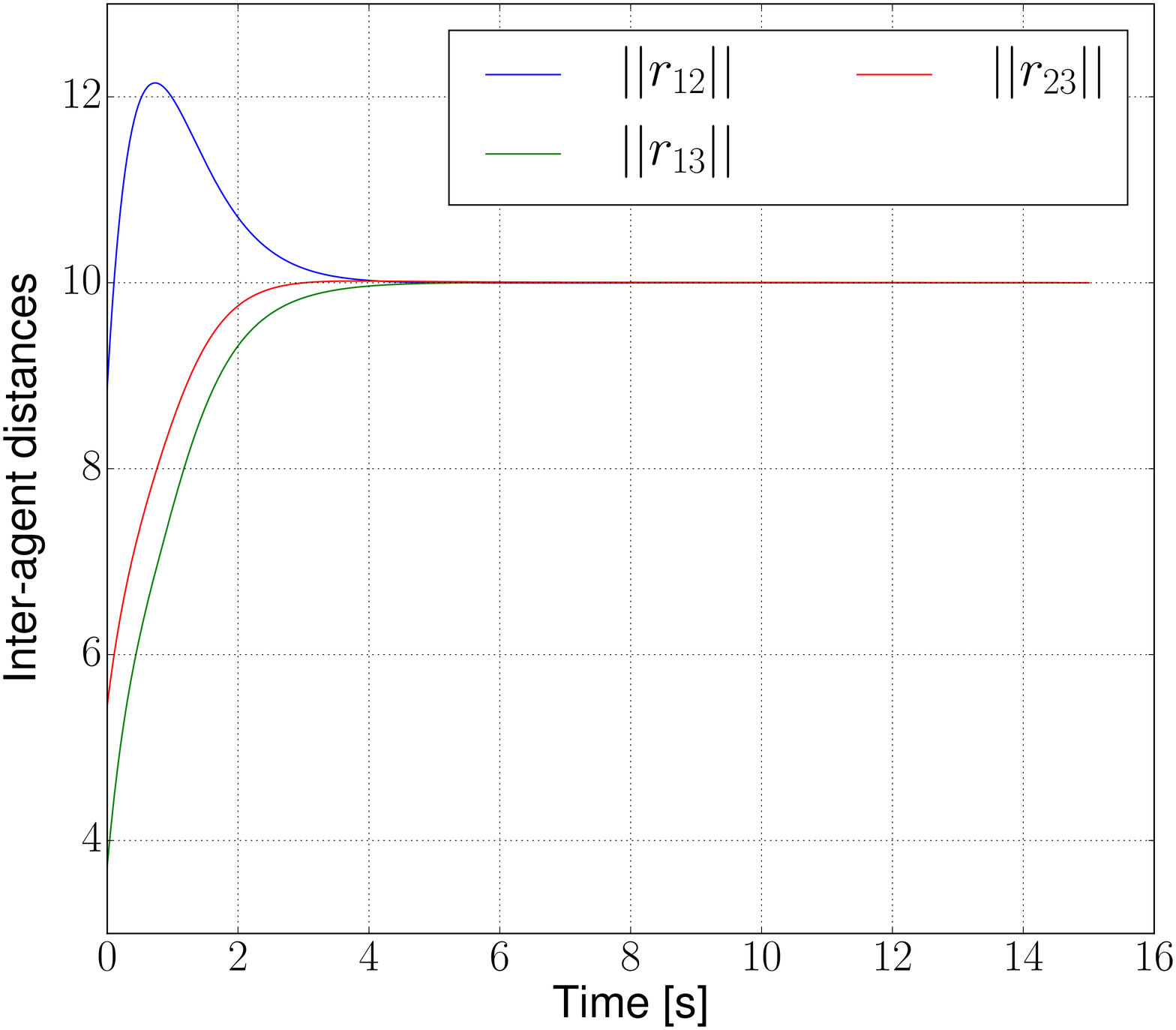}
	\caption{Trajectories described by the agents on the plane. The crosses denote the initial positions, and the dashed-lines describe the steady-state shape. The steady-state rotational motion assists in the correct estimation of the relative positions $r_{ij}$. On the right, the evolution of the actual inter-agent distances converging to the desired value of $10$ for an equilateral triangle.}
	\label{fig: pos}
\end{figure}

\begin{figure}
\centering
	\includegraphics[width=0.6\columnwidth]{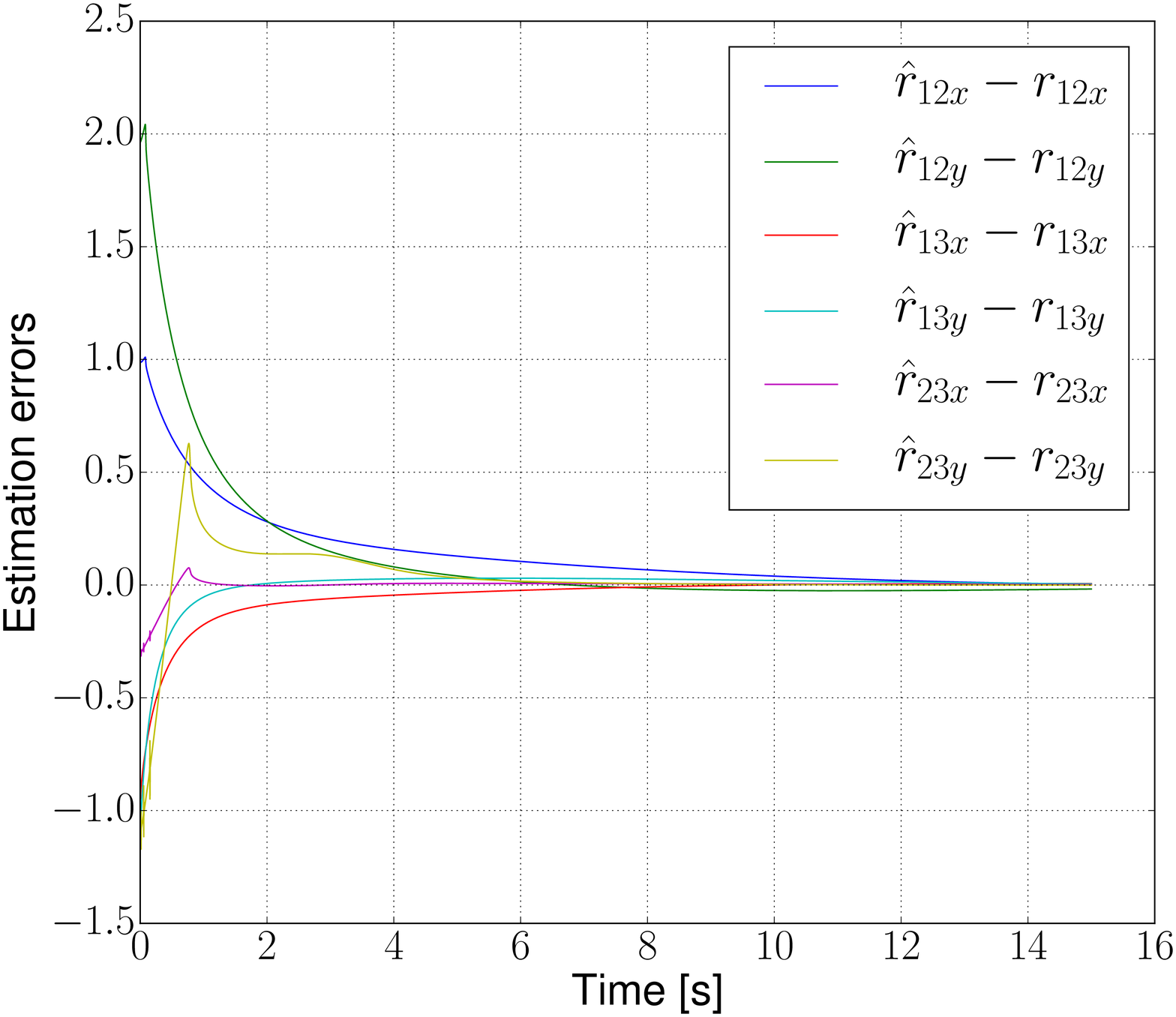}
	\caption{Time evolution of the estimation errors $\hat r_{ij} - r_{ij}$.}
	\label{fig: est}
\end{figure}

\bibliographystyle{IEEEtran} 
\bibliography{hector_ref}

\end{document}